\documentclass[11pt,twoside]{article}
\usepackage{amsmath,amssymb,theorem}
\usepackage{graphicx}

\usepackage[dvips]{epsfig}

\numberwithin{equation}{section}
\newtheorem{theorem}{Theorem}

\newtheorem{remark}{Remark}
\newenvironment{proof}{{\bf Proof}.\ }{ \hfill $\square$}
\newcommand{\R}{\mathbb{R}}

\newcommand{\N}{\mathbb{N}}

\begin{document}

\title{On solutions of the reduced model \\ for the dynamical evolution of contact lines}

\author{Dmitry Pelinovsky \\
{\small Department of Mathematics, McMaster
University, Hamilton, Ontario, Canada, L8S 4K1}}

\date{\today}
\maketitle

\begin{abstract}
We solve the linear advection--diffusion equation with a variable speed on a semi-infinite line.
The variable speed is determined by an additional condition at the boundary, which
models the dynamics of a contact line of a hydrodynamic flow at a $180^{\circ}$ contact angle.
We use Laplace transform in spatial coordinate and Green's function for
the fourth-order diffusion equation to show local existence of solutions of the initial-value
problem associated with the set of over-determining boundary conditions. We also analyze 
the explicit solution in the case of a constant speed (dropping the additional boundary condition).
\end{abstract}

\section{Introduction}

Contact lines are defined by the intersection of the rigid and free boundaries of the flow.
Flows with the contact line at a $180^{\circ}$ contact angle were discussed in \cite{Benney,Dussan},
where corresponding solutions of the Navier--Stokes equations were shown to have no physical meanings.
Recently, a different approach based on the lubrication approximation and thin film equations
was developed by Benilov \& Vynnycky \cite{Benilov}. 

As a particularly simple model for the flow shown on Figure \ref{fig-1}, the authors 
of \cite{Benilov} derived the linear advection--diffusion
equation for the free boundary $h(x,t)$ of the flow:
\begin{equation}
\label{pde}
\frac{\partial h}{\partial t} + \frac{\partial^4 h}{\partial x^4} =
V(t) \frac{\partial h}{\partial x}, \quad x > 0, \;\; t > 0.
\end{equation}
The contact line is fixed at $x = 0$ in the reference frame moving with the velocity $-V(t)$
and is defined by the boundary conditions $h|_{x = 0} = 1$ and $h_x |_{x = 0} = 0$.
The flux conservation condition is expressed by the boundary condition $h_{xxx} |_{x = 0} = -\frac{1}{2}$
(take $\alpha^3 = 3$ in equations (5.12)--(5.13) in \cite{Benilov}).
\begin{figure}[htbp]
\begin{center}
\includegraphics[height=5cm]{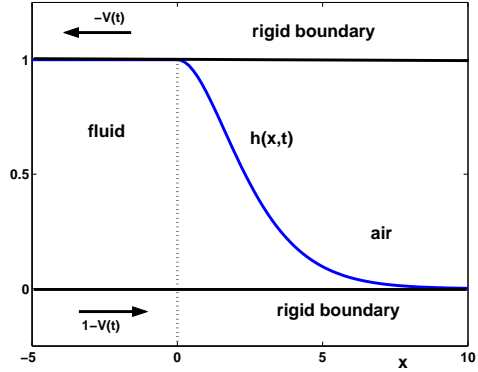}
\end{center}
\caption{Schematic picture of the flow between rigid boundaries.}
\label{fig-1}
\end{figure}

We assume that $h, h_x, h_{xx} \to 0$ as $x \to \infty$: in fact, any constant value of
$h$ at infinity is allowed thanks to the invariance of the linear advection--diffusion 
equation (\ref{pde}) with respect to the shift
and scaling transformations. With three boundary conditions at $x = 0$ and
the decay conditions as $x \to \infty$, the initial-value problem for equation (\ref{pde})
is over-determined and the third (over-determining) boundary condition at $x = 0$ is used
to find the dependence of $V$ on $t$.

We shall consider the initial-value problem with the initial data
$h|_{t = 0} = h_0(x)$ for a suitable function $h_0$. In particular,
we assume that the profile $h_0(x)$ decays monotonically to zero as $x \to \infty$
and that $0$ is a non-degenerate maximum of $h_0$ such that
$h_0(0) = 1$, $h_0'(0) = 0$, and $h_0''(0) < 0$. If the solution $h(x,t)$
losses monotonicity in $x$ during the dynamical
evolution, for instance, due to the value of $h_{xx}(0,t)$ crossing $0$ from the negative side,
then we say that the flow becomes non-physical for further times and
the model breaks. Simultaneously, this may mean that the velocity $V(t)$ blows up, as
it is defined for sufficiently strong solutions of 
the advection--diffusion equation (\ref{pde}) by the contact equation:
\begin{equation}
\label{contact-equation}
h_{xxxxx}(0,t) = V(t) h_{xx}(0,t),
\end{equation}
which follows by differentiation of (\ref{pde}) in $x$ and setting $x \to 0$.

The main claim of \cite{Benilov} based on numerical computations of the reduced equation
(\ref{pde}) as well as more complicated thin-film equations is that for any suitable $h_0$, there is
a finite positive time $t_0$ such that $V(t) \to -\infty$ and $h_{xx}(0,t) \to -0$ as $t \uparrow t_0$.
Moreover, it is claimed that $V(t)$ behaves near the blowup time as the logarithmic function of $t$, e.g.
\begin{equation}
\label{numerical-claim}
V(t) \sim C_1 \log(t_0 - t) + C_2, \quad \mbox{\rm as} \quad t \uparrow t_0,
\end{equation}
where $C_1$, $C_2$ are positive constants.

This paper is devoted to analytical studies of solutions of the advection--diffusion equation (\ref{pde})
and the effects coming from the inhomogeneous boundary condition $h_{xxx}|_{x = 0} = -\frac{1}{2}$
associated with the flux conservation. In particular, we rewrite the evolution equation
for the variable $u = h_x$ in the form
\begin{equation}
\label{pde-u}
u_t + u_{xxxx} = V(t) u_x, \quad \; x > 0, \;\; t > 0,
\end{equation}
subject to the boundary conditions at the contact line
\begin{equation}
\label{bc}
u |_{x = 0} = 0, \quad  u_{xx} |_{x = 0} = -\frac{1}{2}, \quad
u_{xxx} |_{x = 0} = 0, \quad t \geq 0,
\end{equation}
where the boundary conditions $u_{xxx} |_{x = 0} = h_{xxxx} |_{x=0} = 0$ follows
from the boundary conditions $h|_{x = 0} = 1$ and $h_x |_{x =0} = 0$ as well as the original
evolution system (\ref{pde}) as $x \to 0$.

To simplify the problem, we shall also consider the model for given constant $V(t) = V_0$
and drop the third over-determining boundary conditions at the contact line:
\begin{equation}
\label{toy}
\left\{ \begin{array}{l} u_t + u_{xxxx} = V_0 u_x, \quad \quad \quad \quad \quad \; x > 0, \;\; t > 0, \\
u |_{x = 0} = 0, \quad \quad \quad \quad \quad \quad \quad \quad \quad t \geq 0, \\
u_{xx}|_{x = 0} = -\frac{1}{2}, \quad \quad \quad \quad \quad \quad \quad \; t \geq 0.
\end{array} \right.
\end{equation}
Both problems (\ref{pde-u})--(\ref{bc}) and (\ref{toy}) are considered
under the initial condition $u |_{t = 0} = u_0(x)$ with $u_0(0) = 0$, $u_0'(0) < 0$,
and $u_0''(0) = -\frac{1}{2}$, as well as the decay condition $u, u_x, u_{xx} \to 0$ as $x \to \infty$.

Using Laplace transform in spatial coordinate and Green's function for
the fourth-order diffusion equation, we derive an explicit solution
of the boundary-value problem (\ref{toy}). In the case $V_0 = 0$, we show that the
inhomogeneous boundary condition $h_{xxx} |_{x = 0} = u_{xx}|_{x = 0} = -\frac{1}{2}$ leads to
the secular growth of the boundary value $h_{xx} |_{x = 0} = u_x |_{x = 0}$ to positive infinity
as $t \to \infty$. As a result, even if $h_{xx} |_{x=0} < 0$ initially, the
convexity of the solution $h(x,t)$ at the boundary $x = 0$ 
is lost in the finite time. In the case $V_0 < 0$,
we show that no secular growth is observed but the convexity of the solution at the
boundary is still lost in the finite time. Applying the same method, we prove local
existence of solutions of the original boundary-value problem (\ref{pde-u})--(\ref{bc}).
This prepares us to tackle the original conjecture on the finite-time blow-up
in the dynamical behavior of the model, which is still left opened for forthcoming studies.

The remainder of this paper is organized as follows. Section 2 reports explicit
solutions of the boundary--value problem (\ref{toy}) for
$V_0 = 0$ and $V_0 \neq 0$. Section 3 gives the local existence result
for the original problem (\ref{pde-u})--(\ref{bc}). 
Appendix A lists properties of Green's function
for the fourth-order diffusion equation. 

\section{Solution for $V(t) = V_0$}

Because $V(t)$ is nonconstant for the original problem (\ref{pde}),
the Laplace transform in time $t$ is not a useful method for this problem.
On the other hand, since the boundary-value problem is formulated
for half-line, we can use Laplace transform in space $x$:
\begin{equation}
\label{Laplace-x}
U(p,t) = \int_{0}^{\infty} e^{-p x} u(x,t) dx, \quad p > 0.
\end{equation}
We shall develop this method to solve the boundary--value problem (\ref{toy}). 
The explicit solution of this problem will help us to analyze the consequences 
of inhomogeneous boundary condition
$u_{xx}|_{x = 0} = -\frac{1}{2}$ and the constant advection term
$V(t) = V_0$ on the temporal dynamics of the advection--diffusion equation 
with the fourth-order diffusion. 

Let us denote the boundary values:
\begin{equation}
\label{bv-x}
\beta(t) = u_x |_{x=0}, \quad
\gamma(t) = u_{xxx} |_{x=0}.
\end{equation}
Using Laplace transform (\ref{Laplace-x}), we rewrite an evolution problem 
associated with the advection--diffusion equation (\ref{toy}):
\begin{equation}
\label{toy-laplace}
\left\{ \begin{array}{l} U_t + p^4 U - V_0 p U = \gamma(t) - \frac{1}{2}p  + p^2 \beta(t),
\quad t > 0, \\ U |_{t = 0} = U_0(p),\end{array} \right.
\end{equation}
where $U_0$ is the Laplace transform of $u_0 = u|_{t = 0}$. 
By using the variation of parameters, we obtain
\begin{eqnarray}
\nonumber
U(p,t) & = & U_0(p) e^{-t p^4 + t V_0 p} \\
& \phantom{t} & + \int_0^t e^{-(t-s) p^4 + (t-s) V_0 p}
\left( \gamma(s) - \frac{1}{2}p  + p^2 \beta(s) \right) ds.
\label{toy-solution}
\end{eqnarray}
Using the inverse Laplace transform in $x$, we write this solution in the form:
\begin{eqnarray}
\nonumber
u(x,t) & = & \frac{1}{2 \pi i} \int_{c - i \infty}^{c + i \infty}
e^{p x -t p^4 + t V_0 p} \left( \int_0^{\infty} e^{-py} u_0(y) dy \right) dp \\
& \phantom{t} & + \frac{1}{2 \pi i} \int_{c - i \infty}^{c + i \infty} e^{px}
\left( \int_0^t e^{-(t-s) p^4 + (t-s) V_0 p}
\left( \gamma(s) - \frac{1}{2}p  + p^2 \beta(s) \right) ds \right) dp,
\label{toy-solution-inverse}
\end{eqnarray}
where ${\rm Re}(c) > 0$ so that the singularities of the integrand in
the complex $p$-plane remain to the left of the contour of integration.

If $t > 0$ is finite, $u_0 \in L^1(\R_+)$, and $\beta, \gamma \in L^{\infty}_{\rm loc}(\R_+)$,
Fubini's Theorem implies that the integration in $p$ and in $y$, $s$ can be interchanged.
Let us introduce Green's function $G_t(x)$ for the fourth-order diffusion equation
(see Appendix A):
$$
G_t(x) = \frac{1}{2 \pi i} \int_{c + i \infty}^{c + i \infty}
e^{p x - t p^4} dp = \frac{1}{2 \pi} \int_{-\infty}^{\infty} e^{-t k^4 + i k x} dk =
\frac{1}{\pi} \int_0^{\infty} e^{-t k^4} \cos(kx) dk.
$$
Using Green's function, we can rewrite the solution (\ref{toy-solution-inverse}) in the implicit form:
\begin{eqnarray}
\nonumber
u(x,t) & = & \int_0^{\infty} G_t(x+V_0t - y) u_0(y) dy -  \frac{1}{2}
\int_0^t G'_{t-s}(x+V_0(t-s)) ds \\
& \phantom{t} & + \int_0^t \left[ G_{t-s}(x+V_0(t-s)) \gamma(s)
+ G''_{t-s}(x+V_0(t-s)) \beta(s) \right] ds.
\label{solution-implicit}
\end{eqnarray}
The solution is said to be in the implicit form, because the functions $\beta(t)$
and $\gamma(t)$ determined by the boundary conditions (\ref{bv-x}) are not specified yet.

We verify that $\lim_{x \to \infty} u(x,t) = 0$, no matter what $\beta$ and $\gamma$ are,
as long as they are bounded function of $t$. Indeed, by the Lebesgue's Dominated Convergence
Theorem, we have
$$
\int_0^{\infty} G_t(x+V_0t - y) u_0(y) dy \to 0 \quad \mbox{\rm as} \quad x \to \infty
$$
if $u_0 \in L^1(\R)$, because $G_t(x) \to 0$ as $x \to \infty$.
On the other hand, the other three convolution integrals are bounded
if $\beta, \gamma \in L^{\infty}_{\rm loc}(\R_+)$ and $t > 0$ is finite,
because $G_t$, $G_t'$, and $G_t''$ have integrable singularities at $t = 0$. By the same
Lebesgue's Dominated Convergence Theorem, these three integrals decay to zero as $x \to \infty$.

The functions $\beta(t)$ and $\gamma(t)$ are to be found from the integral equations
obtained at the boundary conditions $u(0,t) = 0$ and $u_x(0,t) = \beta(t)$.
These derivations are performed separately for the cases of $V_0 = 0$ and $V_0 \neq 0$.

\subsection{Case $V_0 = 0$}

We rewrite the solution (\ref{solution-implicit}) for $V_0 = 0$:
\begin{equation}
\label{laplace-solution}
u(x,t) = \int_0^{\infty} G_t(x - y) u_0(y) dy -  \frac{1}{2}
\int_0^t G'_{t-s}(x) ds + \int_0^t \left[ G_{t-s}(x) \gamma(s)
+ G''_{t-s}(x) \beta(s) \right] ds.
\end{equation}
Using the boundary values (\ref{bc-Green-1}) and (\ref{bc-Green-2}) for the Greens function
$G_t(x)$ and the boundary condition $u(0,t) = 0$, we evaluate this expression at $x = 0$
and obtain an integral equation for $\beta$ and $\gamma$:
\begin{eqnarray}
\label{gamma}
-\frac{1}{4\pi} \Gamma\left(\frac{1}{4}\right)\int_0^t \frac{\gamma(s)}{(t-s)^{1/4}} ds +
\frac{1}{4\pi} \Gamma\left(\frac{3}{4}\right) \int_0^t \frac{\beta(s)}{(t-s)^{3/4}} ds
& = & \int_0^{\infty} G_t(-y) u_0(y) dy.
\end{eqnarray}

To use the boundary condition $u_x(0,t) = \beta(t)$, we shall recall from equation (\ref{green-self-similar})
that the function $G_t'''(x)$ behaves like $\mathcal{O}(t^{-1})$ for any $x > 0$ and
hence is not integrable in $t$ at $t = 0$. Therefore,
we have to be careful to differentiate the solution in the above convolution form.
The last term of the solution (\ref{laplace-solution}) can be computed 
by using the Fourier transform:
$$
v(x,t):= \int_0^t G''_{t-s}(x) \beta(s) ds = \frac{1}{2\pi} \int_{-\infty}^{+ \infty}
(i k)^2 e^{i k x} \left( \int_0^t e^{-k^4(t-s)} \beta(s) ds \right) dk.
$$
Differentiating this expression in $x$ and integrating by parts in $s$, we obtain
\begin{eqnarray}
\nonumber
v_x(x,t) & = &
\frac{1}{2\pi i} \int_{-\infty}^{+\infty}
(ik)^3 e^{ikx} \left( \int_0^t  e^{-k^4(t-s)} \beta(s) ds \right) dp \\
\nonumber
& = &
\frac{1}{2\pi i} \int_{-\infty}^{+\infty}
\frac{e^{ikx}}{k} \left( \int_0^t \frac{d}{ds} \left(e^{-k^4(t-s)}\right) \beta(s) ds \right) dp \\
\nonumber
& = & \frac{1}{2\pi i} \int_{-\infty}^{+\infty}
\frac{e^{i k x}}{k} \left( \beta(t) - \beta(0) e^{-k^4 t} - \int_0^t e^{-k^4(t-s)} \beta'(s) ds \right) dp \\
& = & \frac{1}{2} \beta(t) - \beta(0) H_t(x) - \int_0^t H_{t-s}(x) \beta'(s) ds. \label{half-residue}
\end{eqnarray}
where
\begin{equation}
H_t(x) := \frac{1}{2\pi i} \int_{-\infty}^{+\infty}
\frac{e^{- t k^4 + i k x}}{k} d k = \frac{1}{\pi} \int_0^{\infty} \frac{e^{-t k^4} \sin(kx)}{k} dk =
\int_0^x G_t(y) dy.
\end{equation}
Here we note that all integrals are evaluated in the principal value sense, because the half-residue at $k = 0$
is canceled out in the resulting expression (\ref{half-residue}). Also we note that the decay of
$v_x$ to zero as $x \to \infty$ is satisfied because of the symmetry and normalization of $G_t$ in (\ref{green-normalization}).
We can now use the boundary conditions (\ref{bc-Green-2}) and $u_x(0,t) = \beta(t)$
to obtain the exact value for $\beta(t)$:
\begin{eqnarray}
\nonumber
\beta(t) & = & 2 \int_0^{\infty} G_t'(-y) u_0(y) dy - \int_0^t G''_{t-s}(0) ds \\
& = & 2 \int_0^{\infty} G_t'(-y) u_0(y) dy + \frac{\Gamma(3/4)}{\pi} t^{1/4}. \label{beta}
\end{eqnarray}
After $\beta(t)$ is found uniquely from (\ref{beta}),
$\gamma(t)$ is found uniquely from the integral equation (\ref{gamma}).
This computation completes the construction of the
exact solution of the boundary--value problem (\ref{toy})
for $V_0 = 0$. Now we turn to the analysis of obtained solution.

\begin{theorem}
Consider the advection--diffusion equation (\ref{toy}) for $V_0 = 0$ with the initial data
$u_0 \in L^1(\mathbb{R}_+)$. Then, there exists a solution $u \in L^{\infty}(\mathbb{R}_+ \times \mathbb{R}_+)$
of the evolution problem in the explicit form (\ref{laplace-solution}), where
$\beta, \gamma \in L^{\infty}_{\rm loc}(\mathbb{R}_+)$ are
defined by (\ref{gamma}) and (\ref{beta}) with $\lim_{t \to \infty} \beta(t) = +\infty$.
\label{theorem-1}
\end{theorem}

\begin{proof}
The convolution integral in the explicit expression (\ref{beta}) can be analyzed from
the Green's function (\ref{green-self-similar}). If $u_0 \in L^1(\mathbb{R})$,
then
$$
\left| \int_0^{\infty} G_t'(-y) u_0(y) dy \right| \leq \frac{\| g' \|_{L^{\infty}} \| u_0 \|_{L^1}}{t^{1/2}},
\quad t > 0.
$$
Therefore, $\beta \in L^{\infty}_{\rm loc}(\mathbb{R}_+)$ and $\beta(t) \sim t^{1/4}$ as
$t \to \infty$ due to the second term in (\ref{beta}). Now, the integral
equation (\ref{gamma}) for $\gamma(t)$ with a weakly singular kernel is well defined
and solutions exist with $\gamma \in L^{\infty}_{\rm loc}(\mathbb{R}_+)$. 
Similarly, the solution $u \in L^{\infty}(\mathbb{R}_+ \times \mathbb{R}_+)$ 
is well defined by (\ref{laplace-solution}).
\end{proof}

\begin{remark}
One can show that there is no singularity of the solution for
$\beta(t)$ as $t \to 0$ so that $\beta(0) = u_0'(0)$ by continuity. Also, one can show
that the solution of the integral equation (\ref{gamma}) for $\gamma(t)$ exists in the closed form:
$\gamma(t) = 2 \int_0^{\infty} G_t(-y) u_0'''(y) dy$.
\end{remark}

Coming back to the original question, if $u_0(0) = 0$, $u_0'(0) < 0$, and $u''_0(0) = -\frac{1}{2}$,
then there is a finite value of $t_0 \in (0,\infty)$
such that $u_x |_{x = 0} > 0$ for all $t > t_0$, that is, $h(x,t)$ loses monotonicity
at $x = 0$ in a finite time $t_0$ (recall that $u = h_x$). This dynamical phenomenon
occurs because of the inhomogeneous boundary conditions $u_{xx} |_{x = 0} = -\frac{1}{2}$
even in the absence of the advection term in the fourth-order diffusion equation (\ref{toy}).

\subsection{Case $V_0 \neq 0$}

We have the solution in the implicit form (\ref{solution-implicit}) and we need
to derive integral equations on the unknown function $\beta(t)$ and $\gamma(t)$.
One integral equation follows again from the boundary condition $u(0,t) = 0$:
\begin{eqnarray}
\nonumber
& \phantom{t} & - \int_0^t \left[ G_{t-s}(V_0(t-s)) \gamma(s) + G''_{t-s}(V_0(t-s)) \beta(s) \right] ds \\
& = & \int_0^{\infty} G_t(V_0t - y) u_0(y) dy -  \frac{1}{2} \int_0^t G'_{t-s}(V_0(t-s)) ds.
\label{int-eq-1}
\end{eqnarray}
To find another integral equation from the boundary condition $u_x(0,t) = \beta(t)$,
we have to use the technique explained in Section 2.1 and to compute the derivative
of the solution (\ref{solution-implicit}) in $x$:
\begin{eqnarray}
\nonumber
u_x(x,t) & = & \int_0^{\infty} G_t'(x+V_0t - y) u_0(y) dy -  \frac{1}{2}
\int_0^t G''_{t-s}(x+V_0(t-s)) ds \\ \nonumber
& \phantom{t} & + \int_0^t G_{t-s}'(x+V_0(t-s)) \gamma(s) ds
+ \frac{1}{2} \beta(t) - \beta(0) H_t(x + V_0t) \\
& \phantom{t} & - \int_0^t H_{t-s}(x + V_0(t-s)) \beta'(s) ds
+ V_0 \int_0^t G_{t-s}(x + V_0(t-s)) \beta(s) ds.
\label{solution-implicit-derivative}
\end{eqnarray}
We can now use the boundary condition $u_x(0,t) = \beta(t)$
to obtain another integral equation for $\beta$ and $\gamma$:
\begin{eqnarray}
\nonumber
& \phantom{t} & \beta(t) + 2 \beta(0) H_t(V_0t) + 2
\int_0^t H_{t-s}(V_0 (t-s)) \beta'(s) ds \\
\nonumber
& \phantom{t} &
- 2 V_0 \int_0^t G_{t-s}(V_0(t-s)) \beta(s) ds - 2 \int_0^t G_{t-s}'(V_0(t-s)) \gamma(s) ds\\
& \phantom{t} & = 2 \int_0^{\infty} G_t'(V_0 t - y) u_0(y) dy -  \int_0^t G''_{t-s}(V_0(t-s)) ds.
\label{int-eq-2}
\end{eqnarray}
The system of integral equations (\ref{int-eq-1}) and (\ref{int-eq-2}) completes the solution
(\ref{solution-implicit}) for the case $V_0 \neq 0$. Because of 
the original motivation to study behavior for large negative $V(t)$ in (\ref{numerical-claim}), 
we shall analyze the obtained solution for $V_0 < 0$.

\begin{theorem}
Consider the advection--diffusion equation (\ref{toy}) for $V_0 < 0$ with the initial data
$u_0 \in L^1(\mathbb{R}_+)$. Then, there exists a solution $u \in L^{\infty}(\mathbb{R}_+ \times \mathbb{R}_+)$
of the evolution problem in the explicit form (\ref{solution-implicit}), where
$\beta, \gamma \in L^{\infty}(\mathbb{R}_+)$ are
defined by (\ref{int-eq-1}) and (\ref{int-eq-2}) with
\begin{equation}
\lim_{t \to \infty} \beta(t) = \frac{1}{2 |V_0|^{1/3}} \quad
\lim_{t \to \infty} \gamma(t) = \frac{|V_0|^{1/3}}{2}.
\label{limits-explicit}
\end{equation}
\label{theorem-2}
\end{theorem}

\begin{proof}
Similarly to the proof of Theorem \ref{theorem-1}, it is easy to show 
from the integral equations (\ref{int-eq-1}) and (\ref{int-eq-2}) 
that if $u_0 \in L^1(\mathbb{R}_+)$, then $\beta, \gamma \in L^{\infty}_{\rm loc}(\mathbb{R}_+)$.
We shall now compute the limit of $\beta(t)$ and $\gamma(t)$ as $t \to \infty$:
\begin{equation}
\label{limits}
\beta_{\infty} := \lim_{t \to \infty} \beta(t), \quad \gamma_{\infty} := \lim_{t \to \infty} \gamma(t).
\end{equation}
To deal with the first integral equation (\ref{int-eq-1}), we first notice the explicit computation
by using the Fourier transform:
\begin{eqnarray*}
f(t) & := & \int_0^t G'_{t-s}(V_0(t-s)) ds \\
& = & \frac{1}{2 \pi}  \int_{-\infty}^{\infty} (ik) \left( \int_0^t e^{-s(k^4 - i k V_0)} ds \right) dk \\
& = & \frac{1}{2 \pi}  \int_{-\infty}^{\infty} \frac{i(1 - e^{-t(k^4 - i k V_0)})}{k^3 - i V_0} dk,
\end{eqnarray*}
where the integrals in $s$ and $k$ can be interchanged by Fubini's Theorem
and the integration is performed in the principal value sense.
We can now explicitly compute the limit as $t \to \infty$ by using
Lebesgue's Dominated Convergence Theorem:
\begin{eqnarray*}
\lim_{t \to \infty} f(t) = \frac{1}{2 \pi}
\int_{-\infty}^{\infty} \frac{i}{k^3 - i V_0} dk =
\frac{-V_0}{\pi}  \int_0^{\infty} \frac{dk}{k^6 + V_0^2} = \frac{1}{3 |V_0|^{2/3}}.
\end{eqnarray*}
This computation gives the last term of the integral equation (\ref{int-eq-1}) as $t \to \infty$.
To deal with the first term on the right-hand side of (\ref{int-eq-1}), we write
\begin{eqnarray*}
\int_0^{\infty} G_t(V_0t - y) u_0(y) dy & = & \frac{1}{2\pi}
\int_{-\infty}^{\infty} \left( \int_0^{\infty} e^{-t(k^4 - i k V_0) - iky} u_0(y) dy \right) dk \\
& = & \int_{-\infty} e^{-t(k^4 - i k V_0)} \hat{u}_0(k) dk,
\end{eqnarray*}
where
$$
\hat{u}_0(k) := \frac{1}{2\pi} \int_0^{\infty} e^{-iky} u_0(y) dy.
$$
By Lebesgue's Dominated Convergence Theorem, this integral converges to zero as $t \to \infty$
as long as $u_0 \in L^1(\R_+)$.

To deal with the second term on the left-hand side of the integral equation (\ref{int-eq-1}), we
rewrite it in the form
$$
\int_0^t G''_{t-s}(V_0(t-s)) \beta(s) ds =
\frac{1}{2 \pi}  \int_{-\infty}^{\infty} (ik)^2 \left( \int_0^t \beta(t-s) e^{-s(k^4 - i k V_0)} ds \right) dk.
$$
Since $\beta \in L_{\rm loc}^{\infty}(\R_+)$ with the assumed limit in (\ref{limits}), we apply
Lebesgue's Dominated Convergence Theorem and compute the integrals
in the principal value sense:
$$
\lim_{t \to \infty} \int_0^t G''_{t-s}(V_0(t-s)) \beta(s) ds =
\frac{-\beta_{\infty}}{2 \pi}  \int_{-\infty}^{\infty} \frac{k}{k^3 - i V_0} dk =
\frac{-\beta_{\infty}}{\pi}  \int_0^{\infty} \frac{k^4 dk}{k^6 + V_0^2} =
\frac{-\beta_{\infty}}{3 |V_0|^{1/3}}.
$$

The first term on the left-hand side of the integral equation (\ref{int-eq-1})
is more tricky. First, we rewrite it in the form,
$$
\int_0^t G_{t-s}(V_0(t-s)) \gamma(s) ds =
\frac{1}{2 \pi}  \int_{-\infty}^{\infty} \left( \int_0^t \gamma(t-s) e^{-s(k^4 - i k V_0)} ds \right) dk.
$$
However, if $\gamma \in L_{\rm loc}^{\infty}(\R_+)$ with the limit in (\ref{limits}),
application of Lebesgue's Dominated Convergence Theorem yields the integral in $k$
with a simple pole at $k = 0$:
$$
\lim_{t \to \infty} \int_0^t G_{t-s}(V_0(t-s)) \gamma(s) ds =
\frac{\gamma_{\infty}}{2 \pi}  \int_{-\infty}^{\infty} \frac{dk}{k (k^3 - i V_0)}.
$$
The integral is no longer understood in the principal value sense. Instead, we return back
to the treatment of the inverse Laplace transform in (\ref{toy-solution-inverse}) with ${\rm Re}(c) > 0$,  
use transformation $o = i k$, and shift the contour of integration in $k$ 
below the pole at $k = 0$. As a result, computations
are completed with the half-residue term at the simple pole and the principal value integral:
\begin{eqnarray*}
\lim_{t \to \infty} \int_0^t G_{t-s}(V_0(t-s)) \gamma(s) ds =
\frac{\gamma_{\infty}}{2 \pi} \left( \frac{\pi}{|V_0|} +
\int_{-\infty}^{\infty} \frac{k^2 dk}{k^6 + V_0^2}  \right) = \frac{2 \gamma_{\infty}}{3 |V_0|}.
\end{eqnarray*}
Combining all computations together, we have obtained the following linear equation
on $\beta_{\infty}$ and $\gamma_{\infty}$ from the integral equation (\ref{int-eq-1}):
\begin{equation}
\label{lin-eq-1-infty}
\frac{2 \gamma_{\infty}}{|V_0|} - \frac{\beta_{\infty}}{|V_0|^{1/3}} = \frac{1}{2 |V_0|^{2/3}}.
\end{equation}

To deal with the second integral equation (\ref{int-eq-2}),
we use the Fourier transform again to write
$$
H_t(V_0 t) = \frac{1}{2\pi i} \int_{-\infty}^{+\infty}
\frac{e^{- t (k^4 - i k V_0)}}{k} d k
$$
and
\begin{eqnarray*}
\int_0^t H_{t-s}(V_0(t-s)) \beta'(s) ds =
\frac{1}{2 \pi i} \int_{-\infty}^{\infty} \frac{1}{k}
\left( \int_0^t \beta'(t-s) e^{-s(k^4-i k V_0)} ds \right) dk,
\end{eqnarray*}
where the integrals are understood in the principal value sense.
If $\beta, \gamma \in L_{\rm loc}^{\infty}$ with the limits (\ref{limits}),
the Lebesgue's Dominated Convergence Theorem implies that
$$
H_t(V_0 t), \int_0^t H_{t-s}(V_0(t-s)) \beta'(s) ds \to 0 \quad \mbox{\rm as} \quad t \to \infty.
$$
Similar to the previous computations, we prove that
\begin{eqnarray*}
\int_0^{\infty} G_t'(V_0t - y) u_0(y) dy \to 0 \quad \mbox{\rm as} \quad t \to \infty, \\
\lim_{t \to \infty} \int_0^t G''_{t-s}(V_0(t-s)) ds = \frac{-1}{3 |V_0|^{1/3}}, \\
\lim_{t \to \infty} \int_0^t G_{t-s}'(V_0(t-s)) \gamma(s) ds =
\frac{\gamma_{\infty}}{3 |V_0|^{2/3}}, 
\end{eqnarray*}
and
\begin{eqnarray*}
\lim_{t \to \infty} \int_0^t G_{t-s}(V_0(t-s)) \beta(s) ds =
\frac{\beta_{\infty}}{2 \pi}  \int_{-\infty}^{\infty} \frac{dk}{k (k^3 - i V_0)}
= \frac{\beta_{\infty}}{6 |V_0|}.
\end{eqnarray*}
where the last integral is computed in the principal value sense because 
equations (\ref{solution-implicit-derivative}) and (\ref{int-eq-2}) 
are derived in the principal value sense. 

Combining all computations together, we have obtained the following linear equation
on $\beta_{\infty}$ and $\gamma_{\infty}$ from the integral equation (\ref{int-eq-2}):
\begin{equation}
\label{lin-eq-2-infty}
4 \beta_{\infty} - \frac{2 \gamma_{\infty}}{|V_0|^{2/3}} = \frac{1}{|V_0|^{1/3}}.
\end{equation}
Solving the linear system (\ref{lin-eq-1-infty}) and (\ref{lin-eq-2-infty}), we obtain
(\ref{limits-explicit}) and the theorem is proved.
\end{proof}

Coming back to the original question, if $u_0(0) = 0$, $u_0'(0) < 0$, and $u''_0 = -\frac{1}{2}$,
then there is a finite value of $t_0 \in (0,\infty)$
such that $u_x |_{x = 0} > 0$ for all $t > t_0$. Therefore,
like in the case $V_0 = 0$, the function $h(x,t)$ loses monotonicity
at $x = 0$ in a finite time $t_0$ (where $u = h_x$) with the only
difference that $u_x |_{x = 0}$ remains finite and positive as $t \to \infty$.
We conclude that the presence of the advection term with $V_0 < 0$
in the fourth-order diffusion equation (\ref{toy})
does not prevent the loss of monotonicity in $x$ but
still stabilizes the solution globally as $t \to \infty$.
In both cases $V_0 = 0$ and $V_0 < 0$, the monotonicity of
$h$ in $x$ is lost because of the inhomogeneous boundary conditions
$h_{xx} |_{x = 0} = -\frac{1}{2}$.

\section{Solution of the original problem}

We shall now use Laplace transform (\ref{Laplace-x}) to
obtain the implicit solution to the advection-diffusion equation (\ref{pde-u})
with a variable speed $V(t)$. Let us denote
$$
W(t) = \int_0^t V(s) ds
$$
and obtain the Laplace transform solution in the form:
\begin{eqnarray}
\nonumber
U(p,t) & = & U_0(p) e^{-t p^4 + W(t) p} \\
& \phantom{t} & + \int_0^t e^{-(t-s) p^4 + (W(t)-W(s)) p}
\left( - \frac{1}{2}p  + p^2 \beta(s) \right) ds.
\label{pde-solution}
\end{eqnarray}
Compared with the solution (\ref{toy-solution}), we have set $\gamma(t) = 0$
because of the third boundary condition in (\ref{bc}).
Using the inverse Laplace transform in $x$ and recalling the definition of the
Green's function $G_t(x)$ (see Appendix A), we obtain the analogue of the implicit solution
(\ref{solution-implicit}):
\begin{eqnarray}
\nonumber
u(x,t) & = & \int_0^{\infty} G_t(x+ W(t) - y) u_0(y) dy - \frac{1}{2}
\int_0^t G'_{t-s}(x+W(t) - W(s)) ds \\
& \phantom{t} & + \int_0^t G''_{t-s}(x+W(t) - W(s)) \beta(s) ds.
\label{solution-implicit-pde}
\end{eqnarray}
Now we have two unknowns $\beta$ and $W$ and we can set up
two integral equations at the boundary conditions
$u(0,t) = 0$ and $u_x(0,t) = \beta(t)$.

From the boundary
condition $u(0,t) = 0$, we obtain the integral equation:
\begin{eqnarray}
\nonumber
-\int_0^t G''_{t-s}(W(t)-W(s)) \beta(s) ds & = &
\int_0^{\infty} G_t(W(t) - y) u_0(y) dy \\
& \phantom{t} & -  \frac{1}{2} \int_0^t G'_{t-s}(W(t)-W(s)) ds.
\label{integral-equation-pde-1}
\end{eqnarray}

To find another integral equation from the boundary condition $u_x(0,t) = \beta(t)$,
we differentiate the solution (\ref{solution-implicit-pde}) in $x$:
\begin{eqnarray}
\nonumber
u_x(x,t) & = & \int_0^{\infty} G_t'(x+ W(t) - y) u_0(y) dy -  \frac{1}{2}
\int_0^t G''_{t-s}(x+W(t)-W(s)) ds \\
\nonumber & \phantom{t} &
+ \frac{1}{2} \beta(t) - \beta(0) H_t(x + W(t))
- \int_0^t H_{t-s}(x + W(t) - W(s)) \beta'(s) ds \\
& \phantom{t} & + V(t) \int_0^t G_{t-s}(x + W(t)- W(s)) \beta(s) ds.
\label{solution-implicit-derivative-pde}
\end{eqnarray}
From the boundary condition $u_x(0,t) = \beta(t)$,
we obtain another integral equation:
\begin{eqnarray}
\nonumber
& \phantom{t} & \beta(t) + 2 \beta(0) H_t(W(t)) + 2
\int_0^t H_{t-s}(W(t) - W(s)) \beta'(s) ds \\
\nonumber
& \phantom{t} &
- 2 V(t) \int_0^t G_{t-s}(W(t)-W(s)) \beta(s) ds \\
& \phantom{t} & = 2 \int_0^{\infty} G_t'(W(t) - y) u_0(y) dy -  \int_0^t G''_{t-s}(W(t)-W(s)) ds.
\label{integral-equation-pde-2}
\end{eqnarray}

We shall prove that the system of two integral equations
(\ref{integral-equation-pde-1}) and (\ref{integral-equation-pde-2})
determines uniquely the function $\beta(t)$ and $V(t)$ locally for $t > 0$.
The following theorem gives the result.

\begin{theorem}
Assume that $u_0 \in C^{\infty}(\R_+)$ such that
\begin{equation}
\label{initial-value-u0}
u_0(0) = 0, \quad u_0''(0) = -\frac{1}{2}, \quad
u_0'''(0) = 0.
\end{equation}
Then, there exists a formal solution $(V,\beta)$ of the system of two integral equations
(\ref{integral-equation-pde-1}) and (\ref{integral-equation-pde-2})
in the form of the fractional power series:
\begin{equation}
\label{series-fractional}
\beta(t) = \beta_0 + \sum_{n=4}^{\infty} \beta_{n/4} t^{n/4}, \quad
V(t) = V_0 + \sum_{n=1}^{\infty} V_{n/4} t^{n/4},
\end{equation}
where $\beta_0 = u_0'(0)$, $V_0 = u_0^{(4)}(0)/u_0'(0)$, and $\{ \beta_{n/4}, V_{(n-3)/4} \}_{n = 4}^{\infty}$
are uniquely determined.
\end{theorem}

\begin{proof}
We substitute the series representations (\ref{series-fractional}) to each term of the integral equations
(\ref{integral-equation-pde-1}) and (\ref{integral-equation-pde-2}). It follows
from (\ref{series-fractional}) that
$$
a_t = \frac{1}{t^{1/4}} \int_0^t V(s) ds = V_0 t^{3/4} + \sum_{n=1}^{\infty} \frac{4}{n+4} V_{n/4} t^{(n+3)/4}
$$
and
$$
\xi_{t,\tau} = \frac{1}{\tau^{1/4}} \int_{t-\tau}^t V(s) ds = V_0 \tau^{3/4}
+ \sum_{n=1}^{\infty} \frac{4}{n+4} V_{n/4} \frac{t^{(n+4)/4}-(t-\tau)^{(n+4)/4}}{\tau^{1/4}}.
$$

Using the representation (\ref{green-self-similar}) of the Green function with $g \in C^{\infty}(\R)$,
we obtain for the three terms of the integral equation (\ref{integral-equation-pde-1}):
\begin{eqnarray*}
& \phantom{t} & \int_0^t G''_{t-s}(W(t)-W(s)) \beta(s) ds = \beta_0 \int_0^t \frac{g''(\xi_{t,\tau})}{\tau^{3/4}} d \tau +
\sum_{n=4}^{\infty} \beta_{n/4} \int_0^t \frac{g''(\xi_{t,\tau}) (t-\tau)^{n/4}}{\tau^{3/4}} d \tau \\
& = & 4 \beta_0 g''(0) t^{1/4} + \sum_{k=2}^{\infty} \frac{1}{k!} g^{(k+2)}(0) \int_0^t \frac{\xi^k_{t,\tau}}{\tau^{3/4}} d \tau +
\sum_{n=4}^{\infty} \beta_{n/4} \int_0^t \frac{g''(\xi_{t,\tau}) (t-\tau)^{n/4}}{\tau^{3/4}} d \tau,
\end{eqnarray*}
\begin{eqnarray*}
& \phantom{t} & \int_0^{\infty} G_t(W(t) - y) u_0(y) dy = \int_0^{\infty} g(z-a_t) u_0(t^{1/4} z) dz =
\sum_{n=1}^{\infty} \frac{1}{n!} u_0^{(n)}(0) t^{n/4} \int_0^{\infty} g(z-a_t) z^n dz \\
& = & t^{1/4} u_0'(0) \sum_{k=0}^{\infty} \frac{1}{k!} (-a_t)^k \int_0^{\infty} g^{(k)}(z) z dz
+ \sum_{n=2}^{\infty} \frac{1}{n!} u_0^{(n)}(0) t^{n/4} \int_0^{\infty} g(z-a_t) z^n dz,
\end{eqnarray*}
and
\begin{eqnarray*}
\int_0^t G'_{t-s}(W(t)-W(s)) ds = \int_0^t \frac{g'(\xi_{t,\tau})}{\tau^{2/4}} d \tau =
\sum_{k=1}^{\infty} \frac{1}{k!} g^{(k+1)}(0) \int_0^t \frac{\xi^k_{t,\tau}}{\tau^{2/4}} d \tau.
\end{eqnarray*}

At the first powers of $t^{1/4}$, we obtain a system of linear algebraic equations on
the coefficients of the power series (\ref{series-fractional}):
\begin{eqnarray*}
t^{1/4} : & \phantom{t} & \quad -4 \beta_0 g''(0) = u_0'(0) \int_0^{\infty} g(z) z dz, \\
t^{2/4} : & \phantom{t} & \quad 0 = \frac{1}{2!} u_0''(0) \int_0^{\infty} g(z) z^2 dz, \\
t^{3/4} : & \phantom{t} & \quad 0 = \frac{1}{3!} u_0'''(0) \int_0^{\infty} g(z) z^3 dz, \\
t^{4/4} : & \phantom{t} & \quad 0 = \frac{1}{4!} u_0^{(4)}(0) \int_0^{\infty} g(z) z^4 dz - u_0'(0) V_0 \int_0^{\infty} g'(z) z dz, \\
t^{5/4} : & \phantom{t} & \quad -\beta_{4/4} g''(0) \int_0^1 \frac{(1-x)^{4/4}}{x^{3/4}} dx =
\frac{1}{5!} u_0^{(5)}(0) \int_0^{\infty} g(z) z^5 dz - \frac{1}{2} u_0''(0) V_0 \int_0^{\infty} g'(z) z^2 dz \\
& \phantom{t} & \quad \quad \quad \quad \quad \quad \quad \quad \quad \quad \quad \quad
- \frac{4}{5} u_0'(0) V_{1/4} \int_0^{\infty} g'(z) z dz -\frac{2}{5} g''(0) V_0,
\end{eqnarray*}
and so on. Using the explicit values for the integrals (\ref{integral1})--(\ref{integral5})
and the initial conditions (\ref{initial-value-u0}), we obtain $\beta_0 = u_0'(0)$,
$V_0 = u_0^{(4)}(0)/u_0'(0)$, and the linear equation
\begin{equation}
\label{lin-eq-1}
u_0'(0) V_{1/4} + 8 g''(0) \left(\beta_{4/4} + u_0^{(5)}(0) + \frac{1}{2} V_0 \right) = 0
\end{equation}

Similarly, we work with the terms of the second integral equation (\ref{integral-equation-pde-2}):
\begin{eqnarray*}
& \phantom{t} & \int_0^t G_{t-s}(W(t)-W(s)) \beta(s) ds = \beta_0 \int_0^t \frac{g(\xi_{t,\tau})}{\tau^{1/4}} d \tau +
\sum_{n=4}^{\infty} \beta_{n/4} \int_0^t \frac{g(\xi_{t,\tau}) (t-\tau)^{n/4}}{\tau^{1/4}} d \tau \\
& = & \frac{4}{3} \beta_0 g(0) t^{3/4} + \sum_{k=2}^{\infty} \frac{1}{k!} g^{(k)}(0) \int_0^t \frac{\xi^k_{t,\tau}}{\tau^{1/4}} d \tau +
\sum_{n=4}^{\infty} \beta_{n/4} \int_0^t \frac{g(\xi_{t,\tau}) (t-\tau)^{n/4}}{\tau^{1/4}} d \tau,
\end{eqnarray*}
\begin{eqnarray*}
& \phantom{t} & \int_0^{\infty} G_t'(W(t) - y) u_0(y) dy = \int_0^{\infty} g(z-a_t) u'_0(t^{1/4} z) dz =
\sum_{n=0}^{\infty} \frac{1}{n!} u_0^{(n+1)}(0) t^{n/4} \int_0^{\infty} g(z-a_t) z^n dz \\
& = & u_0'(0) \sum_{k=0}^{\infty} \frac{1}{k!} (-a_t)^k \int_0^{\infty} g^{(k)}(z) dz
+ \sum_{n=1}^{\infty} \frac{1}{n!} u_0^{(n+1)}(0) t^{n/4} \int_0^{\infty} g(z-a_t) z^n dz,
\end{eqnarray*}
\begin{eqnarray*}
\int_0^t G''_{t-s}(W(t)-W(s)) ds = \int_0^t \frac{g''(\xi_{t,\tau})}{\tau^{3/4}} d \tau =
\sum_{k=0}^{\infty} \frac{1}{k!} g^{(k+2)}(0) \int_0^t \frac{\xi^k_{t,\tau}}{\tau^{3/4}} d \tau,
\end{eqnarray*}
\begin{eqnarray*}
H_t(W(t)) = \int_0^{a_t} g(z) dz =
\sum_{k=0}^{\infty} \frac{1}{(k+1)!} g^{(k)}(0) a_t^{k+1},
\end{eqnarray*}
and
\begin{eqnarray*}
\int_0^t H_{t-s}(W(t) - W(s)) \beta'(s) ds = \sum_{n=4}^{\infty} \beta_{n/4}
\sum_{k=0}^{\infty} \frac{1}{(k+1)!} g^{(k)}(0)
\int_0^t \xi_{t,\tau}^{k+1} (t-\tau)^{(n-4)/4} d \tau
\end{eqnarray*}

At the first powers of $t^{1/4}$, we obtain a system of linear algebraic equations on
the coefficients of the power series (\ref{series-fractional}):
\begin{eqnarray*}
t^{0/4} : & \phantom{t} & \quad \beta_0 = 2 u_0'(0) \int_0^{\infty} g(z) dz, \\
t^{1/4} : & \phantom{t} & \quad 0 = 2 u_0''(0) \int_0^{\infty} g(z) z dz - 4g''(0), \\
t^{2/4} : & \phantom{t} & \quad 0 = u_0'''(0) \int_0^{\infty} g(z) z^2 dz, \\
t^{3/4} : & \phantom{t} & \quad 2 \beta_0 g(0) V_0 -\frac{8}{3} \beta_0 g(0) V_0 =
\frac{1}{3} u_0^{(4)}(0) \int_0^{\infty} g(z) z^3 dz - 2 u_0'(0) V_0 \int_0^{\infty} g'(z) dz, \\
t^{4/4} : & \phantom{t} & \quad \beta_{4/4} + \frac{8}{5} \beta_0 g(0) V_{1/4} - \frac{8}{3} \beta_0 g(0) V_{1/4} =
\frac{1}{12} u_0^{(5)}(0) \int_0^{\infty} g(z) z^4 dz - 2 u_0''(0) V_0 \int_0^{\infty} g'(z) z dz \\
& \phantom{t} & \quad \quad \quad \quad \quad \quad \quad \quad \quad \quad \quad \quad
- \frac{8}{5} u_0'(0) V_{1/4} \int_0^{\infty} g'(z) dz,
\end{eqnarray*}
and so on. Again, using the explicit values for the integrals (\ref{integral1})--(\ref{integral5})
and the initial conditions (\ref{initial-value-u0}), we obtain $\beta_0 = u_0'(0)$,
$V_0 = u_0^{(4)}(0)/u_0'(0)$, and the linear equation
\begin{equation}
\label{lin-eq-2}
-\frac{8}{3} u_0'(0) g(0) V_{1/4} + \left(\beta_{4/4} + u_0^{(5)}(0) + \frac{1}{2} V_0 \right) = 0
\end{equation}
The system of linear equations (\ref{lin-eq-1}) and (\ref{lin-eq-2}) has a unique solution
\begin{equation}
\label{lin-eq-sol}
V_{1/4} = 0, \quad \beta_{4/4} = - u_0^{(5)}(0) - \frac{1}{2} V_0,
\end{equation}
provided that
$$
-\frac{64}{3} g(0) g''(0) = \frac{4}{3 \pi^2} \Gamma\left(\frac{1}{4}\right)\Gamma\left(\frac{3}{4}\right)
= \frac{4 \sqrt{2}}{3 \pi} \neq 1,
$$
which is confirmed. Note that the
constraint $V_0 = u_0^{(4)}(0)/u_0'(0)$ also follows from the 
contact equation (\ref{contact-equation}) obtained 
for sufficiently smooth solutions. Similarly,
the second equation (\ref{lin-eq-sol}) follows from the advection--diffusion 
equation (\ref{pde-u})
after one derivative in $x$ is taken in the limit $x \to 0$ and $t \to 0$.

It remains to prove that the system
of linear equations obtained from the system of integral equations
(\ref{integral-equation-pde-1}) and (\ref{integral-equation-pde-2})
can be solved at each order of $t^{n/4}$ for $n \in \N$. From the previous computations,
we can deduce that the first integral equation at $t^{(n+1)/4}$ gives a
linear equation on variables $(\beta_{n/4},V_{(n-3)/4})$ of the power series
(\ref{series-fractional}):
\begin{equation}
\label{lin-eq-1-new}
-\beta_{n/4} g''(0) \int_0^1 \frac{(1-x)^{n/4}}{x^{3/4}} dx +
\frac{4}{n+1} u_0'(0) V_{(n-3)/4} \int_0^{\infty} g'(z) z dz = \cdots,
\end{equation}
where the dots denote the terms expressed through derivatives of $u_0(x)$ at $x  = 0$
and the previous terms of the power series (\ref{series-fractional}). Similarly,
the second integral equation at $t^{n/4}$ gives another linear equation
on variables $(\beta_{n/4},V_{(n-3)/4})$:
\begin{equation}
\label{lin-eq-2-new}
\beta_{n/4} - \frac{8}{3} \beta_0 g(0) V_{(n-3)/4} = \cdots.
\end{equation}
The system of linear equations (\ref{lin-eq-1-new}) and (\ref{lin-eq-2-new}) is non-degenerate
if
\begin{equation}
\label{C-n}
C_n := -\frac{4(n+1)}{3} g(0) g''(0) \int_0^1 \frac{(1-x)^{n/4}}{x^{3/4}} dx
= \frac{(n+1)}{6 \sqrt{2} \pi}
\frac{\Gamma\left(\frac{n+4}{4}\right)\Gamma\left(\frac{1}{4}\right)}{\Gamma\left(\frac{n+5}{4}\right)}
\neq 1, \quad n \in \N.
\end{equation}
The left-hand side $\{ C_n \}_{n \in \N}$ is computed numerically (see Figure \ref{fig-2}). 
It is a monotonically
increasing sequence that approaches closely to $1$ at $n = 8$, where $C_8 \approx 0.96$,
and $n = 9$, where $C_9 \approx 1.04$. Therefore, the linear system is non-degenerate
and a unique solution for $(\beta_{n/4},V_{(n-3)/4})$ exists for any $n \in \N$.
\end{proof}
\begin{figure}[htbp]
\begin{center}
\includegraphics[height=5cm]{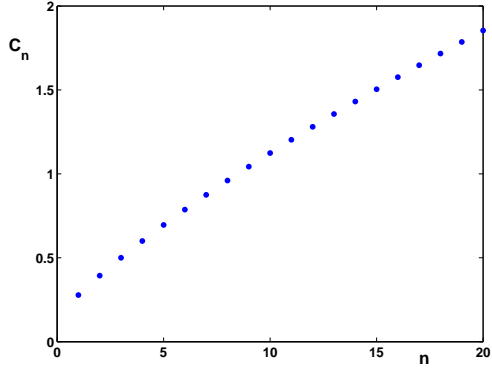}
\end{center}
\caption{Numerical approximations of $C_n$ defined by (\ref{C-n}).}
\label{fig-2}
\end{figure}

In the present time, we cannot prove yet that the system of integral equations
(\ref{integral-equation-pde-1}) and (\ref{integral-equation-pde-2}) leads to a finite-time
blow-up, according to the conjecture in \cite{Benilov}. Nevertheless, numerical computations
show that the blow-up holds for a generic set of initial data.
Figure \ref{fig-3} shows the behavior of functions $\beta(t)$ and $V(t)$ near the blow-up time. 
It follows from this figure that $\beta(t) = h_{xx}|_{x = 0} \to 0$
at the same time as $V(t) \to -\infty$ with $\beta(t) V(t)^{1/3} \to C_0$, where
$C_0 > 0$ is a numerical constant. In other words, we conclude with
the conjecture that $\beta(t) \sim V(t)^{-1/3}$ as $V(t) \to -\infty$ in a finite time $t_0 \in (0,\infty)$.
\begin{figure}[htbp]
\begin{center}
\includegraphics[height=6.5cm]{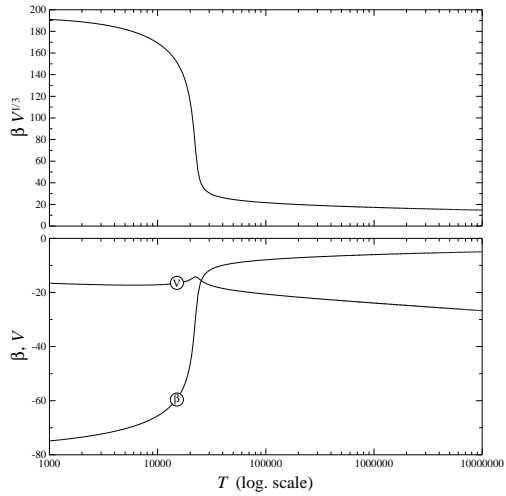}
\end{center}
\caption{Numerical computations of $\beta(t)$ and $V(t)$ for
the advection-diffusion equation (\ref{pde}). We thank 
the authors of \cite{Benilov} for this numerical figure.}
\label{fig-3}
\end{figure}

\appendix

\section{Green's function}

Let us define the fundamental solution of the fourth-order diffusion equation:
\begin{equation}
\label{Green}
\left\{ \begin{array}{l} h_t + h_{xxxx} = 0, \quad x \in \R, \;\; t > 0, \\
h |_{t = 0} = \delta(x), \quad \quad x \in \R, \end{array} \right.
\end{equation}
where $\delta$ is a standard Dirac delta-function in the distribution sense.
The fundamental solution is usually referred to as Green's function and we shall denote it
by
$$
h(x,t) = G_t(x), \quad x \in \R, \quad t \in \R_+.
$$

Using the Fourier transform in $x$, we can obtain the explicit expression for 
Green's function:
\begin{equation}
\label{Green-Fourier}
G_t(x) = \frac{1}{2\pi} \int_{-\infty}^{\infty} e^{-t k^4 + i k x} dk =
\frac{1}{\pi} \int_0^{\infty} e^{-t k^4} \cos(kx) dk.
\end{equation}
In particular, we have $G_t(-x) = G_t(x)$ for all $x \in \R$ and
\begin{eqnarray}
\label{bc-Green-1}
& \phantom{t} & G_t(0) = \frac{1}{\pi} \int_0^{\infty} e^{-t k^4} dk =
\frac{\Gamma(1/4)}{4 \pi t^{1/4}}, \\\label{bc-Green-2}
& \phantom{t} & G_t''(0) = -\frac{1}{\pi} \int_0^{\infty} k^2 e^{-t k^4} dk
= -\frac{\Gamma(3/4)}{4 \pi t^{3/4}},
\end{eqnarray}
where $\Gamma$ is the standard Gamma function.
The Green's function can be represented in the self-similar form by
\begin{equation}
\label{green-self-similar}
G_t(x) = \frac{1}{t^{1/4}} g\left(\frac{x}{t^{1/4}}\right), \quad
g(z) = \frac{1}{\pi} \int_0^{\infty} e^{-k^4} \cos(kz) dk,
\end{equation}
where $g \in L^2(\R) \cap L^{\infty}(\R)$. Therefore, $G_t$ decays
to zero as $t \to \infty$ in any $L^p$ norm for $p \geq 2$.
In particular, $| G_t(x) | \leq \| g \|_{L^{\infty}}/t^{1/4}$,
$| G_t'(x) | \leq \| g' \|_{L^{\infty}}/t^{1/2}$, and so on, for any $x \in \R$.

By the stationary phase method (see, e.g., Chapter 5 in \cite{Miller}),
$g(z)$ and all derivatives of $g(z)$ decay to zero as $|z| \to \infty$
faster than any algebraic powers. This gives the decay of
$G_t(x)$ and any $x$-derivative of $G_t(x)$ as $|x| \to \infty$
for any fixed $t > 0$. Although $G_t$ and $g$ are not $L^1$ functions, they satisfy
the normalization conditions:
\begin{equation}
\label{green-normalization}
\int_{\R} G_t(x) dx = \int_{\R} g(z) dz = 1, \quad t > 0.
\end{equation}

The even function $g : \R \to \R$ satisfies the ordinary differential equation
\begin{equation}
4 \frac{d^4 g}{d z^4} = g + z \frac{d g}{d z}, \quad z \in \R,
\end{equation}
subject to the initial values
\begin{equation}
g(0) = \frac{1}{4 \pi} \Gamma\left(\frac{1}{4}\right), \quad g'(0) = 0, \quad
g''(0) = -\frac{1}{4 \pi} \Gamma\left(\frac{3}{4}\right), \quad g'''(0) = 0,
\end{equation}
and the decay behavior as $|z| \to \infty$. It is clear from
the differential equation that $g \in C^{\infty}(\R)$ 
satisfies a number of integral constraints:
\begin{eqnarray}
\label{integral1}
\int_0^{\infty} z g(z) dz & = & -4 g''(0), \\
\int_0^{\infty} z^2 g(z) dz & = & 0, \\
\int_0^{\infty} z^3 g(z) dz & = & -8 g(0), \\
\int_0^{\infty} z^4 g(z) dz & = & -12, \\
\label{integral5}
\int_0^{\infty} z^5 g(z) dz & = & 16 4! g''(0),
\end{eqnarray}
and so on.

\end{document}